 \documentclass[final,3p,times]{elsarticle}

\usepackage{amssymb,amsmath,amsthm,bm}
\usepackage{url}
\usepackage{multirow,bigstrut}
\usepackage{graphicx}
\usepackage{subfigure}
\usepackage{amsopn}
\usepackage{enumerate}
\usepackage{geometry}
\usepackage{algorithm}
\usepackage{algorithmic}
\usepackage{multirow, booktabs}
\usepackage{caption}
\usepackage[normalem]{ulem}
\usepackage[retainorgcmds]{IEEEtrantools}

\newlength\figwidth
\newtheorem{Theorem}{Theorem}
\DeclareMathOperator\Prob{Prob}
\makeatletter
\newcommand{\dotminus}{\mathbin{\text{\@dotminus}}}
\newcommand{\@dotminus}{
  \ooalign{\hidewidth\raise1ex\hbox{.}\hidewidth\cr$\m@th-$\cr}}

\journal{xxx}

\begin{document}

\begin{frontmatter}



\title{On the security of a class of diffusion mechanisms for image encryption}


\author[cityu]{Leo Yu Zhang\corref{corr}}
\ead{leocityu@gmail.com}
\author[xtu]{Yuansheng Liu}
\author[cityu]{Kwok-Wo Wong}
\author[ENDIF]{Fabio Pareschi}
\author[swu]{Yushu Zhang}
\author[DEI]{Riccardo Rovatti}
\author[ENDIF]{Gianluca Setti}
\cortext[corr]{Corresponding author.}
\address[cityu]{Department of Electronic Engineering, City University of Hong Kong, Hong Kong, China}
\address[xtu]{School of Software, Dalian University of Technology, Dalian, China}
\address[ENDIF]{Engineering Department in Ferrara, University of Ferrara, Italy}
\address[swu]{School of Electronics and Information Engineering, Southwest University, Chongqing, China}
\address[DEI]{Department of Electrical, Electronic and Information Engineering, University of Bologna, Italy}

\begin{abstract}
The need for fast and strong image cryptosystems motivates researchers to develop new techniques to apply traditional cryptographic primitives in order to exploit the intrinsic features of digital images. One of the most popular and mature technique is the use of complex dynamic phenomena, including chaotic orbits and quantum walks, to generate the required key stream. In this paper, under the assumption of plaintext attacks we investigate the security of a classic diffusion mechanism (and of its variants) used as the core cryptographic primitive in some image cryptosystems based on the aforementioned complex dynamic phenomena. We have theoretically found that regardless of the key schedule process, the data complexity for recovering each element of the equivalent secret key from these diffusion mechanisms is only $O(1)$. The proposed analysis is validated by means of numerical examples. Some additional cryptographic applications of our work are also discussed.
\end{abstract}


\end{frontmatter}


\section{Introduction}
\label{sec:intro}
The recent years increase in the popularity of the internet and multimedia communication has resulted in the fast development of information exchange and consumer electronics applications.
However, it has also led to an increase in the demand of secure and real-time transmission of these data.
The easiest way to cope with this is to consider the multimedia stream as a standard bit stream and apply traditional cryptographic approaches like 3DES \cite{barker2012sp} and AES \cite{daemen2002design} with proper mode of operation.
Yet, the desire for cryptosystems more efficient and specifically designed for multimedia stream has drawn increasing research attention in the past decade
\cite{liu2010survey,lian2013design,lian2007commutative,cheng2000partial,li2007design,
magli2011transparent,zeng2003efficient,fridrich1998symmetric,chen2004symmetric,mao2004novel,zhang2014chaotic,
parvin2014new,norouzi2014simple,yang2015novel,zhang2013symmetric,zhu2011chaos}.
A particular field of interest in this area is the development of strong and fast image cryptosystems.

Two major approaches can be identified in the literature for the design of image encryption algorithms.
The first one exploits some complex dynamic phenomena, such as chaotic behavior and quantum walks,
as the image encryption algorithm core.
Many schemes belonging to this approach are based on the permutation-diffusion architecture depicted in Fig.~\ref{fig:structure}, which was first proposed by Fridrich in \cite{fridrich1998symmetric}.
The encryption process is based on the iteration of permutation (i.e., image element transposition)
and diffusion (i.e., value modification) operations.
Almost all works proposing an extension of Fridrich's work can be categorized into the following two classes:

\begin{enumerate}\setlength{\itemsep}{0pt}
\item Developing novel permutation techniques. In Fridrich's original design, permutation is implemented
 by iterating a 2D discretized chaotic map like Baker or Cat map. Chen \textit{et al.} suggested using 3D
 chaotic map to de-correlate the relationship among pixels in a more efficient way \cite{chen2004symmetric,mao2004novel}.
 In \cite{wong2008fast}, Wong \textit{et al.} proposed an ``add-and-then-shift" strategy to include certain
 amount of diffusion effect into permutation, thus reducing the overall number of iteration rounds, and improving the efficiency. For the same purpose, Zhu \textit{et al.} suggested carrying out permutation to bit-level
 instead of pixel-level \cite{zhu2011chaos,zhang2013symmetric}. It is also worth mentioning that there are permutation
 techniques based on general Gray code \cite{zhou2013n,zanin2014gray}, which can be considered as permutation
 carried out at an arbitrary bit length.

\item Developing novel diffusion techniques. As illustrated by Fridrich in \cite{fridrich1998symmetric}, the diffusion
 operation aims to spread the information of plaintext to the whole ciphertext. This process can be formulated as
    \begin{equation*}
    c(l) = p(l) \dotplus G(c(l-1), k(l)),
    \end{equation*}
 where $\dotplus$ denotes the modulo addition, $p(l)$, $c(l)$ and $k(l)$ denote the $l$-th plaintext element, ciphertext
 element and element derived from the secret key, respectively. For security and efficiency considerations, the function
 $G$ should be both simple and nonlinear,
 a typical example is a chaos-based look-up table \cite{wong2002fast}. By taking advantage of
 the low complexity and non-commutable properties between the bitwise exclusive or and the modulo addition operation,
 which are popular in traditional crytosystems like IDEA and RC6, Chen \textit{et al.} in \cite{chen2004symmetric} suggested
 implementing diffusion according to the following formula
    \begin{equation}
    \begin{IEEEeqnarraybox}[][c]{rCl}
    \IEEEstrut
    c(l) = ( p(l) \dotplus k(l)) \oplus k(l) \oplus c(l-1), \nonumber
    \IEEEstrut
    \end{IEEEeqnarraybox}
    \end{equation}
    where $\oplus$ stands for bitwise exclusive or.
    Many other works adopt similar (or even the same) diffusion mechanisms, see \cite{zhu2014image,zhu2012novel,zhang2013symmetric,yang2015novel,Rao:ModifiedCKBA:ICDSP07,
    parvin2014new,norouzi2014simple,CH:HCKBA:IJBC10,liu2011color,zhou2015cascade} for examples.
    It is not surprising that the computational efficient modulo multiplication can also be incorporated into the diffusion stage \cite{zhu2014image,wang2014cryptanalysis}. Moreover, recent works suggested using real number arithmetic to enhance the security level of the diffusion stage \cite{norouzi2014simple,yang2015novel}
    at the cost of a reduced computational efficiency due to the employment of complicated arithmetic operations.
\end{enumerate}

\begin{figure}[htbp]
\centering
\includegraphics[width=3.5in]{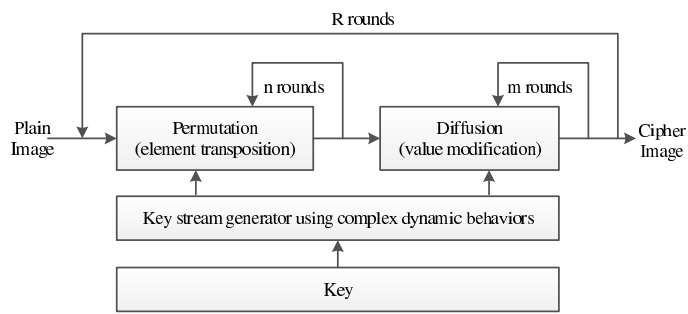}
\caption{Schematic diagram of Fridrich's permutation-diffusion architecture.}
\label{fig:structure}
\end{figure}

The second major approach in the design of image cryptosystem is based on optical technology schemes, which are
supposed to benefit from the intrinsic property of optic systems to process high dimensional complex data in parallel.
The most classic image cryptosystem based on optical technology is the double random phase encoding (DRPE) method
developed by R\'{e}fr\'{e}gier and Javidi in \cite{refregier1995optical}. A comprehensive review on this topic
can be found in \cite{chen2014advances}.
Though the DRPE technique has several advantages, like high speed,
multidimensional processing and robustness, the underlying arithmetic operation, which is matrix multiplication, is
linear. From the cryptanalysis point of view, linearity leads to a low security level. Thus the DRPE method is
vulnerable under various kinds of attack \cite{peng2006known,carnicer2005vulnerability,peng2006chosen} and the adoption of
image cryptosystem based on optical technology for real application should be cautious.

In this paper we take into account the first approach only, i.e., that exploiting complex dynamic phenomena. In particular, we investigate on some security-related aspects of these systems. Note that in any image cryptosytem, security is a critical issue. In fact, due to the particular structure of digital image files (such as, for example, horizontal/vertical correlation) many statistical analysis based methods may reduce the security.
Typical statistical tests include histogram analysis, correlation analysis, entropy analysis \cite{wu2013local}, sensitivity analysis \cite{chen2004symmetric}
and randomness analysis \cite{Rukhin:TestPRNG:NIST}.

In recent years, a lot of image ciphers employing complex dynamic phenomena and fulfilling all the aforementioned statistical tests requirements, have been proposed but afterwards found to be insecure under various attack models
\cite{arroyo2013cryptanalysis,li2011breaking,li2013breaking,li2012breaking,
Zhang2015breaking,wang2014cryptanalysis,solak2010cryptanalysis}. For example, the equivalent key stream used for
permutation of Fridrich's design can be retrieved in chosen-plaintext (CP) attack scenario \cite{solak2010cryptanalysis}
and a chaos-based image cipher with Feistel structure is insecure with respect to differential attack when the round number
is smaller than $5$ \cite{zhang2012cryptanalyzing}.
Note that in the literature, the cryptanalysis of these image ciphers is usually performed case-by-case, since any cryptanalytic method is usually effective only on a particular image cipher.
Conversely, despite being more useful from a theoretical point of view, only a few works provide security evaluation of some general cryptographic components.
In \cite{li2008general}, Li \textit{et al.} presented a general quantitative study of permutation-only
encryption algorithms against plaintext attacks. Their result was further improved in \cite{Jolfaei2015TIFS} with respect to data and computation complexity.
In \cite{chen2012period,chen2013period,chen2014period},
Chen \textit{et al.} studied the period distribution of the generalized discrete Cat map,
which is a fundamental building block in many permutation schemes.

In this paper we want to make a step further in the evaluation of generic cryptographic components for image cryptosystem by studying the security of the differential equation of modulo addition (DEA) in the form $(\alpha \dotplus k) \oplus (\beta \dotplus k) = y$. This analysis is not completely new.
In \cite{li2011breaking}, it was reported that $3$ pairs of chosen queries $(\alpha, \beta)$ are sufficient to reveal the unknown $k$ of the formula $(\alpha \dotplus k) \oplus (\beta \dotplus k) = y$. It is further reduced to 2 pairs of chosen
$(\alpha, \beta)$ in \cite{li2013breaking}.
As far as we know, these works must be considered as independent analyses of particular image ciphers
\cite{CH:HCKBA:IJBC10,Rao:ModifiedCKBA:ICDSP07}. In our previous work \cite{Zhang2015breaking},
it was reported that the diffusion mechanism suggested by Chen \textit{et al.} \cite{chen2004symmetric} can be cast
to the form $(\alpha \dotplus k) \oplus (\beta \dotplus k) = y$ under CP attack and the similar method can be also used
to analyze other DEA that includes modulo multiplication operation.

In more detail, we take into account the three image cryptosystems proposed in  \cite{parvin2014new},
in \cite{norouzi2014simple} and in \cite{yang2015novel} as case studies, all of them adopting Fridrich's
permutation-diffusion scheme, and we study the resistance against plaintext attack of the adopted diffusion mechanisms by exploiting security results achieved by the aforementioned DEA equation analysis.
Specifically, we evaluate the data complexity (i.e., required number of pairs of $(\alpha, \beta)$) for solving
$(\alpha \dotplus k) \oplus (\beta \dotplus k) = y$ and its extension in a known-plaintext (KP) attack scenario.
The main difference between this work and previous ones is that we assume that $\alpha$ and $\beta$ cannot be freely chosen, as for example in \cite{li2011breaking,li2013breaking}. This allows us to apply obtained results to the security analysis of the three aforementioned cryptosystem schemes. A full analytic result is presented to derive a sufficient condition for solving the equation $(\alpha \dotplus k) \oplus (\beta \dotplus k) = y$; furthermore, some design weakness of its variants are pointed out. Numerical simulation results are then provided to support our analyses.


The innovative contribution of this paper is three-fold. First, we analyze the relationship between
a class of popular diffusion mechanisms
and the DEA $(\alpha \dotplus k) \oplus (\beta \dotplus k) = y$ by studying three example image ciphers
\cite{parvin2014new,norouzi2014simple,yang2015novel}. It is also worth mentioning that the similar DEA can be found in many other designs \cite{zhu2014image,zhu2012novel,zhang2013symmetric,yang2015novel,Rao:ModifiedCKBA:ICDSP07,
    parvin2014new,norouzi2014simple,CH:HCKBA:IJBC10,liu2011color,zhou2015cascade}
     so the application of our analyses is not limited to the three case studies.
Second, we analytically investigate the sufficient
condition to solve $(\alpha \dotplus k) \oplus (\beta \dotplus k) = y$ and we also experimentally
present a simple KP attack to  a variant of this DEA. The conclusion drawn from our result is that security is substantially lower than the desired one.
Third, we study the three encryption schemes \cite{parvin2014new,norouzi2014simple,yang2015novel} which combines the investigated diffusion mechanism and secret random permutation. Their security is evaluated in detail.

The rest of this paper is organized as follows. Section~\ref{sec:notations} introduces the notations that is used in this paper
and the assumptions we work on. The three image cryptosystem case studies are reviewed in Sec.~\ref{sec:threeciphers} and the
differential equations of modulo addition are derived in Sec.~\ref{sec:problem}. Section~\ref{sec:mainresult} presents security
analyses and numerical results of the equations derived above against KP attack. The applications of our results are discussed
in Sec.~\ref{sec:applications} and conclusion remarks are drawn in the last section.

\section{Notations and main assumptions}
\label{sec:notations}
In the following, we will use the notation $\{p(i,j)\}_{i=1,j=1}^{H, W}$ and $\{p(k)\}_{k=1}^{L}$ to represent the $2$D
and $1$D format of a plain-image of size $L = H \times W$ (Height $\times$ Width). The $2$D and $1$D representations of the cipher-image $C$
are $\{c(i,j)\}_{i=1,j=1}^{H, W}$ and $\{c(k)\}_{k=1}^{L}$, respectively. We use $a_i$ to denote the $i$-th bit of an $n$-bit integer $a$ ($a \in \mathbb{Z}_2^n$) and $(a_{n-1}\cdots a_0)_2$ to denote the binary form of $a$. The default value of $n$ is $8$ unless otherwise specified.
The symbols `$\dotplus$', `$\dotminus$',`$\oplus$', `$\wedge$' and `$\|$' denote \textit{modulo $2^n$ addition}, \textit{modulo $2^n$ subtraction}, \textit{bitwise exclusive or} (XOR), \textit{bitwise and} and \textit{bitwise or}, respectively. We will use $ab$ to represent $a\wedge b$
and $\lfloor x \rfloor$ ($\lceil x \rceil$) to represent the largest (smallest) integer not greater (less) than the real number $x$.
The cardinality of a set $A$ is denoted by $\#{A}$. With the term $KS$ we will refer to all the key schedule operations of a specific algorithm, and use $KS(Seed)$ to indicate the process generating all necessary key streams given a secret
 $Seed$
by the $KS$.

In order to correctly evaluate the security level of a diffusion mechanism either in known- or chosen-plaintext attack scenario, we clarify here the power of the adversary. In the KP attack model, the adversary has access to some plaintexts and their corresponding ciphertexts. In the CP attack model, we assume that the adversary can obtain ciphertexts from any plaintext of his choice. In both scenarios, the goal of the attack is either to collect information on the secret key $Seed$ or, equivalently, on the key stream(s) $KS(Seed)$ generated from $Seed$. Hereinafter, we will consider only the problem of recovering $KS(Seed)$.


\section{Image cryptosystems review}
\label{sec:threeciphers}
In this section, we briefly review the three cryptosystems for image encryption proposed in \cite{parvin2014new}, in \cite{norouzi2014simple}, and in \cite{yang2015novel}. A detailed description of the three schemes can be found in the original works\footnote{For the sake of both clarity and uniformity, some notations and/or some operations may have been changed without affecting the security level of the schemes.}. Here, we want to highlight that, though the key schedule process of these schemes are different from the each other, all of the schemes share a very similar diffusion mechanism in the encryption process. In the next section, we will exploit this to cast the three diffusion mechanisms into the same general form
and evaluate their cryptographic strength.


\begin{enumerate}[A.]
\item \textbf{Parvin's cryptosystem.} The key schedule operation of the cipher proposed in \cite{parvin2014new}
    is based on two chaotic functions and the encryption process is composed by a row/column circular permutation and a
    sequential pixel diffusion.
        \begin{enumerate}[1]
        \item{\textit{Initialization}:} Generate three key streams $U = \{u(i)\}_{i=1}^{H}$, $V = \{v(i)\}_{i=1}^{W}$ and $K = \{k(i)\}_{i=0}^{L}$
        from $KS(Seed)$, where $U$, $V$ and $K$ are composed of random integers in interval $[1, W]$, $[1, H]$ and $[0, 255]$, respectively.
        \item{\textit{Permutations}:} Carry out row circular permutation to the plain-image $P$ using
        \begin{equation}
        p'(i, (j+u(i)) \bmod W) = p(i, j),
        \label{eq:parvin1}
        \end{equation}
        and denote the result by $P'$. Then permute $P'$ further using the circular column permutation as follows
        \begin{equation}
        s((i+v(j)) \bmod H, j) = p'(i, j).
        \label{eq:pavin2}
        \end{equation}
        \item{\textit{Diffusion}:}  Stretch $S$ to a $1$D sequence $\{s(l)\}_{l=1}^{L}$ and calculate the pixel values of the
        cipher-image by the following diffusion equation
        \begin{equation}
        c(l) = s(l)  \oplus (c(l-1) \dotplus k(l) ) \oplus k(l),
        \label{eq:diffusion}
        \end{equation}
        where $l \in [1, 2, \cdots, L]$ and $c(0) = k(0)$. Rearrange $\{c(l)\}_{l=1}^{L}$ to a matrix of size $H \times W$
        to get the cipher-image $C$.
        \end{enumerate}

\item \textbf{Norouzi's cryptosystem.} The key schedule suggested in \cite{norouzi2014simple} is based on
the hyper-chaotic system introduced in \cite{yujun2010new}. The encryption process is composed by a single diffusion process,
which can be viewed as the generalized version of the previous diffusion scheme.

    \begin{enumerate}[1]
        \item{\textit{Initialization}:} Produce a key stream $K = \{k(i)\}_{i=0}^{L}$ by running $KS(Seed)$, where $k(i)$ is $8$-bit
            integer in $[0, 255]$.
        \item{\textit{Diffusion}:} Calculate the pixel values of the cipher-image sequentially by the following bidirectional diffusion equation
        \begin{equation}
        c(l) = p(l)  \oplus (c(l-1) \dotplus k(l) ) \oplus f(P, k(l)),
        \label{eq:diffusion2}
        \end{equation}
        where $l \in [1, 2, \cdots, L]$, $c(0) = k(0)$ and
        \begin{equation}
        f(P, k(l)) = \lfloor (\sum\nolimits_{i=l+1}^{L} p(i)) \cdot k(l) \cdot 10^8 /256^4 \rfloor \bmod 256.
        \label{eq:prediffusion2}
        \end{equation}
        Rearrange $\{c(l)\}_{l=1}^{L}$ to a matrix of size $H \times W$ and denote it as $C$.
    \end{enumerate}

\item \textbf{Yang's cryptosystem.} The key schedule of the image cryptosystem proposed in \cite{yang2015novel} is
derived from the one-dimensional two-particle discrete-time quantum random walks, which is totally different from those suggested in \cite{parvin2014new,norouzi2014simple}. However, the encryption process, which is composed of a diffusion stage and a permutation
stage, is an extension of Norouzi's work \cite{norouzi2014simple}.

    \begin{enumerate}[1]
        \item{\textit{Initialization}:} Obtain the key streams $K = \{k(i)\}_{i=0}^{L}$, $U = \{u(i)\}_{i=1}^{W}$ and $V = \{v(i)\}_{i=1}^{H}$
        by running the key schedule $KS(Seed)$, where $K$ is composed of $8$-bit integers in the interval $[0, 255]$ and $U$
        and $V$ are permutation of the set $\{1, 2, \cdots, W\}$ and $\{1, 2, \cdots, H\}$, respectively.
        \item{\textit{Diffusion}:} Run the bidirectional diffusion technique characterized by Eq.~(\ref{eq:diffusion2}) to the plain-image pixels
        as follows
        \begin{equation}
        p'(l) = p(l)  \oplus (p'(l-1) \dotplus k(l) ) \oplus f(P, k(l)),
        \label{eq:Yang1}
        \end{equation}
        where $l \in [1, 2, \cdots, L]$, $p'(0) = k(0)$ and $f(P, k(l))$ is defined by Eq.~(\ref{eq:prediffusion2}).
        Rearrange $\{p'(l)\}_{l=1}^{L}$ to a matrix of size $H \times W$ and denote it as $P'$.
        \item{\textit{Permutations}:} Permute the intermediate result $P'$ using the key streams $U$ and $V$ and get the cipher-image $C$, i.e.,
        \begin{eqnarray}
        s(i, u(j))  &=& p'(i, j), \label{eq:Yang2}\\
        c(v(i), j)  &=& s(i, j), \label{eq:Yang3}
        \end{eqnarray}
        where $i \in [1, H]$ and $j \in [1, W]$.
    \end{enumerate}

\end{enumerate}

\section{Problem formulation}
\label{sec:problem}
The cryptosystems shown in the previous section are based either on a single round permutation-diffusion architecture
(Parvin's and Yang's cipher) or on a bidirectional diffusion stage (Norouzi's cipher).
In this paper, we focus our attention on the security of the considered diffusion schemes in a plaintext attack. To this aim, we will neglect at this moment all the effects of the permutation schemes in \cite{parvin2014new,norouzi2014simple,yang2015novel},
that will be considered in Sec.~\ref{sec:applications} only, along with the security of the whole cryptosystems. Mathematically, we assume that all elements of the key streams $U$ and $V$ used for permutation in Parvin's cryptosystem are zeros, and that $U$ and $V$ in Yang's cryptosystem are both given by the identity permutation. Note that a similar approach, with a general quantitative plaintext attack on permutation-only ciphers can be found in \cite{li2008general}.

In the diffusion mechanism proposed by Parvin we will show that the problem of finding the key stream $K$ used in the diffusion scheme with a KP attack is equivalent to solve the DEA $(\alpha \dotplus k) \oplus (\beta \dotplus k) = y$, where $\alpha, \beta, y$ are known parameters and $k$ is unknown. Note that the same DEA, under the assumption that $\alpha$ and $\beta$ can be freely chosen, have already been analyzed by other works, that are also briefly reviewed. We will also show that also the problem of retrieving the key stream for diffusion in Norouzi and Yang's design under CP attack scenario is equivalent to solve this DEA.
In addition, we will also investigate the security level of the diffusion approach proposed by Norouzi and Yang with respect to a KP attack.

\subsection{Parvin's diffusion scheme}
\label{subsec:4.1}
In Parvin's scheme, we assume that two plain-images, $P_1$ and $P_2$, and their corresponding cipher-images, $C_1$ and $C_2$, are available.
Referring to Eq.~(\ref{eq:diffusion}), we have
\begin{equation}
\left\{
\begin{IEEEeqnarraybox}[][c]{rCl}
\IEEEstrut
 c_1(l) &=& p_1(l)  \oplus (c_1(l-1) \dotplus k(l) ) \oplus k(l) \nonumber \\
 c_2(l) &=& p_2(l)  \oplus (c_2(l-1) \dotplus k(l) ) \oplus k(l), \nonumber
\IEEEstrut
\end{IEEEeqnarraybox}
\right.
\end{equation}
where $l \in [1, L]$. Their difference can be calculated as
\begin{equation}
(c_1(l-1) \dotplus k(l) ) \oplus (c_2(l-1) \dotplus k(l))   = c_1(l) \oplus   c_2(l)  \oplus p_1(l) \oplus p_2(l) .
\label{eq:deduction}
\end{equation}
More generally, we can recast this expression by observing that for any value of $l$ we have 
\begin{equation}
(\alpha \dotplus k) \oplus (\beta \dotplus k) = y.
\label{eq:model1}
\end{equation}
In the present context, the problem of finding the key stream $\{k(l)\}_{l=1}^{L}$ of Parvin's cryptosystem
is turned into solving 
Eq.~(\ref{eq:model1}) under some pairs of known parameters $(\alpha, \beta, y)$. Note that $k(0)$,
and so the full strem $K$, can be easily calculated according to Eq.~(\ref{eq:diffusion}) after $\{k(l)\}_{l=1}^{L}$
are revealed.

It is already known that, under the assumption that $\alpha$ and $\beta$ can be chosen freely, $k$
can be determined by only two groups of chosen queries by referring to the following
\begin{Theorem}\cite[Proposition 3 and Corallary 3.1]{li2013breaking}
Suppose $\alpha, \beta, k, y \in \mathbb{Z}_2^n$ and $n>2$, two groups of chosen queries $(\alpha, \beta)$ and their corresponding
$y$ are sufficient to determine $k$ of the following equation
\begin{equation*}
(\alpha \dotplus k) \oplus (\beta \dotplus k) = y
\end{equation*}
in terms of modulo $2^{n-1}$. Specifically the two chosen queries can be $(\hat{\alpha}, \hat{\beta})= (\sum\nolimits_{j=0}^{\lceil n/2 \rceil -1} (00)_2\cdot 4^j), \sum\nolimits_{j=0}^{\lceil n/2 \rceil -1} (10)_2\cdot 4^j)$
and $(\bar{\alpha}, \bar{\beta})= (\sum\nolimits_{j=0}^{\lceil n/2 \rceil -1} (10)_2\cdot 4^j, \sum\nolimits_{j=0}^{\lceil n/2 \rceil -1} (01)_2\cdot 4^j)$.
\end{Theorem}
The proof of Theorem~1 can be found in \cite{li2011breaking,li2013breaking}, and an interpretation from the
computational point of view about this theorem can be found in \cite{Zhang2015breaking}. It is worth mentioning
that the most significant bit (MSB) of $k$, i.e., $k_{n-1}$,
cannot be determined even with additional queries of $(\alpha, \beta)$.
This is intrinsic in the fact that the carry bit generated by the highest bit plane is discarded after
the modulo operation \cite{Zhang2015breaking}. Consequently, both $k$ and $\hat{k} = k \oplus 2^{n-1}$ are two
equivalent solutions of the considered equation. For this reason, in the following we consider only the
problem of determining the $(n-1)$ least significant bits (LSBs) of $k$ in Eq.~(\ref{eq:model1}).

Note however that, by referring to Eq.~(\ref{eq:deduction}), neither a KP nor a CP attack scenario allows us to choose the value
of $\alpha$ and $\beta$ since they represent ciphertext elements. In order get a result similar to that of Theorem~1 that can be applied to the considered cryptosystems, we systematically analyze Eq.~(\ref{eq:model1})
 in Sec.~\ref{subsec:5.1} under the assumption that $\alpha$ and $\beta$ are known to the attacker but cannot be freely chosen.

\subsection{Norouzi and Yang's diffusion scheme}
In Norouzi's and Yang's cryptosystems, the diffusion stage is characterized by Eq.~(\ref{eq:diffusion2}), where some
computational-intensive operations are added to the XOR and modulo addition. Regardless of their computational efficiency,
we are curious whether this new diffusion mechanism will improve the security of the resultant cryptosystem. Given a plain-image
$P_1=\{p_1(l)\}_{l=1}^{L}$, we define the real number sequence $T_1=\{t_1(l)\}_{l=1}^{L}$ as
\begin{equation}
t_1(l) = \sum\nolimits_{i=l+1}^{L} p_1(i)/ 256^4.
\label{eq:tempvalue}
\end{equation}
Then, the diffusion scheme characterized by Eq.~(\ref{eq:diffusion2}) can be written as
  \begin{equation}
        c_1(l) = p_1(l)  \oplus (c_1(l-1) \dotplus k(l) ) \oplus g(t_1(l), k(l)),
        \label{eq:multiply}
  \end{equation}
where $g(t_1(l), k(l)) = \lfloor t_1(l) \cdot (10^8 \cdot k(l)) \rfloor \bmod 256$.
Under a CP attack scenario, an adversary can choose another plain-image $P_2$, which differs from $P_1$ by a single pixel at location $l_0$. In this way the real number sequence $T_2=\{t_2(l)\}_{l=1}^{L}$ associated to $P_2$ satisfies
\begin{equation*}
t_2(l) = t_1 (l) \textit{~~~if~} l \geq l_0.
\end{equation*}
Referring to Eq.~(\ref{eq:multiply}), it is easy to observe that the difference between
$C_1$ and $C_2$ at location $l_0$ will satisfy
\begin{IEEEeqnarray}{rCl}
c_1(l_0) \oplus c_2(l_0) \oplus  p_1(l_0) \oplus  p_2(l_0) &=&  (c_1(l_0-1) \dotplus k(l_0) ) \oplus g(t_1(l_0), k(l_0)) \nonumber\\
&&  \oplus(c_2(l_0-1) \dotplus k(l_0) ) \oplus g(t_2(l_0), k(l_0)) \nonumber \\
  &=& (c_1(l_0-1) \dotplus k(l_0) ) \oplus (c_2(l_0-1) \dotplus k(l_0) ),\nonumber
\end{IEEEeqnarray}
which coincides exactly with Eq.~(\ref{eq:model1}). In conclusion, under the CP attack scenario, the problem of finding
the equivalent secret key stream for diffusion of Norouzi and Yang's designs is converted into solving Eq.~(\ref{eq:model1}) with
some pairs of known parameters $(\alpha, \beta, y)$.

Conversely, under the assumption of a KP attack scenario, we can observe from Eq.~(\ref{eq:tempvalue}) that the calculation of the real number sequence $T$ is independent of the secret key (stream). Then, limiting ourselves to consider the plain
image $P_1$, we can recast Eq.~(\ref{eq:multiply}) as
\begin{equation}
 (\alpha \dotplus k) \oplus g(\beta, k)=y ,
\label{eq:model2}
\end{equation}
where $g(\beta, k) = \lfloor \beta \cdot (10^8 \cdot k) \rfloor \bmod 256$ is a nonlinear function. The problem of determining
$k$ for Eq.~(\ref{eq:model2}) from some groups of known $(\alpha, \beta, y)$ is considered in Sec.~\ref{subsec:5.2}.
Here, special attention should be paid to the fact that $\beta$ is no longer $8$-bit integer but a non-negative real number.

\section{Main results}
\label{sec:mainresult}
\subsection{Cryptographic strength of the equation $(\alpha \dotplus k) \oplus (\beta \dotplus k) = y$}
\label{subsec:5.1}
According to Sec.~\ref{subsec:4.1}, both KP and CP attacks to Parvin¡¯s diffusion scheme are equivalent to solve Eq.~(\ref{eq:model1}) under the assumption that the value of $\alpha$, $\beta$ and $y$ are known but none of them can be chosen.
In the ideal case, the data complexity for to determine $k$ should be $2^{2n}$ because there are $2^{2n}$ possible combinations of
$\alpha$ and $\beta$ in total. However, we can theoretically show (and we will confirm this with simulation results) that the actual complexity substantially deviates from the ideal one.

Let us assume that an adversary successfully collects a set of known triples $(\alpha, \beta, y)$ and denote this set by
\begin{equation*}
\mathbb{G} = \{(\alpha, \beta, y) \mid y = (\alpha \dotplus k) \oplus (\beta \dotplus k)\}
\end{equation*}
with $\# \mathbb{G}=g$.
The candidate solutions of $k$ given $\mathbb{G}$ can be computed by means of a brute-force search according to the following algorithm whose computational complexity is $O(2^{n-1}\cdot g)$.
\begin{itemize}
\item{Step~(1)} Let $l=1$ and the solution set $\mathbb{K}_l = \emptyset$.
\item{Step~(2)} Select the $l$-th element of $\mathbb{G}$ and exhaustively test all the $2^{n-1}$ possible values of $k$
(the MSB of $k$ is ignored here) to check whether it satisfies Eq.~(\ref{eq:model1}). Collect all the possible values of $k$
that meet the requirement and denote them as $\mathbb{K}_l$.
\item{Step~(3)} Set $l=l+1$ if $l< g$. Go to Step~(2) and update the solution set by $\mathbb{K}_{l+1} = \mathbb{K}_{l+1} \cap \mathbb{K}_{l}$.
\end{itemize}

This algorithm ends up with a solution set $\mathbb{K}_g$ which contains all the possible values of $k$ that are
consistent with the known parameter set $\mathbb{G}$. Nevertheless, it is concluded that the computational complexity
is $O(2^{n-1}\cdot g)$ steps. Nevertheless, this algorithm has two shortcomings:
1) there is no hint on how to choose the correct $k$ from $\mathbb{K}_g$ if $\# \mathbb{K}_g \geq 2$;
2) the efficiency is not satisfactory when $n$ is large.
In the case of Parvin's cryptosystem, $n$ is fixed to $8$, and this makes this algorithm working pretty well.
However, in the scheme proposed in \cite{CH:HCKBA:IJBC10,Rao:ModifiedCKBA:ICDSP07}, where $n=32$, this
algorithm becomes inefficient. These two questions are solved on the basis of Theorem~2, where the sufficient condition
to determine the bit plane of $k$ is given.
\begin{Theorem}
Suppose $\alpha, \beta, k, y \in \mathbb{Z}_2^n$ and $n\geq 2$.
Given $\alpha, \beta$ and $y$, the $i$ least significant bits $(0 \leq i < n-1)$ of $k$ of the following equation
\begin{equation*}
(\alpha \dotplus k) \oplus (\beta \dotplus k) = y
\end{equation*}
can be solely determined if $y=\sum\nolimits_{j=0}^{i-1} 2^j = {\overbrace{(0\dots0  \underbrace{1\cdots11}_{i}  )}^{MSB \leftarrow LSB} }\,_2$.
\end{Theorem}
\begin{proof}
The proof of this theorem can be found in \ref{sec:proof}.
\end{proof}

For a given known parameter triple $(\alpha, \beta, y)$, Theorem~2 states that some least significant bits of $k$ can be
confirmed when consecutive ones are observed at the LSBs of $y$.
A more surprising inference drawn from Theorem~2 is that Eq.~(\ref{eq:model1}) can be solved using only a single query
$(\alpha, \beta)$ when the adversary obtains the oracle machine outputs ($2^n-1$) or ($2^{n-1}-1$).

Furthermore, it is also easy to conclude that the result given by Theorem~1 is just a special case of that
by Theorem~2. In detail, for the two chosen queries used in Theorem 1, we have
\begin{IEEEeqnarray}{rCl}
\hat{y} \| \bar{y} &=& (\hat{\alpha} \dotplus k) \oplus (\hat{\beta} \dotplus k) \|  (\bar{\alpha} \dotplus k) \oplus (\bar{\beta} \dotplus k) \nonumber\\
                     &=& 2^n -1 \nonumber
\end{IEEEeqnarray}
and we can also indicate other two groups of queries satisfying the requirements of Theorem 1, specifically
$(\tilde{\alpha}, \tilde{\beta})=(\sum\nolimits_{j=0}^{\lceil n/2 \rceil -1} (10)_2\cdot 4^j), \sum\nolimits_{j=0}^{\lceil n/2 \rceil -1} (00)_2\cdot 4^j)$
and $(\check{\alpha}, \check{\beta})=(\sum\nolimits_{j=0}^{\lceil n/2 \rceil -1} (00)_2\cdot 4^j, \sum\nolimits_{j=0}^{\lceil n/2 \rceil -1} (01)_2\cdot 4^j)$.
Based on Theorem~2, we propose the following efficient algorithm to get a candidate solution of $k$ from the known parameters
set $\mathbb{G}$, with $\#\mathbb{G} =g$.
\begin{itemize}
\item{Step~(1)} Generate parameter sets $\mathbb{G}_{j}  \subseteq \mathbb{G}$ using the following rule
\begin{equation*}
\mathbb{G}_j = \{(\alpha, \beta, y) \mid y = (\alpha \dotplus k) \oplus (\beta \dotplus k),~y_j = 1\},
\end{equation*}
where $j = 0 \sim n-2$. 
\item{Step~(2)} Let $i=0$, $c_0 = 0$ and set the default value of $k$ to a random number in $[0, 2^n-1]$.
\item{Step~(3)}  Refresh the $i$-th bit $k_i$ by look up Table~\ref{table:carrybits} if $\# \mathbb{G}_i\neq 0$
and then calculate $c_{i+1}$ by Eq.~(\ref{eq:bitform}).
\item{Step~(4)} If $i<n-2$, increase $i$ by $1$. Go to Step~(3) if $\# \mathbb{G}_i\neq 0$.
\item{Step~(5)} Calculate $k$ using the equation $k = \sum\nolimits_{i=0}^{n-1}k_i \cdot 2^i$.
\end{itemize}

\begin{table}[htp!]
\centering
\caption{The values of $k_i$ corresponding to the values of $\alpha_i, \beta_i, c_i, y_{i}$, and $\tilde{y}_{i+1}$.}
\centering
{\begin{tabular}{*{8}{c}c}
\toprule
\multirow{2}{*}{$(y_{i}, \tilde{y}_{i+1})$}     & \multicolumn{8}{c}{$(\alpha_i, \beta_i, c_{i} )$} \\
\cline{2-9} & $(0,0,0)$ & $(1,0,0)$ & $(0,1,0)$ & $(0,0,1)$ & $(1,1,0)$ & $(1,0,1)$ & $(0,1,1)$ & $(1,1,1)$\\\midrule
(0, 0)      &   0, 1    &   0, 1    &   -    &   0, 1    &   0, 1    &   -       &  0, 1 &   0, 1    \\
(0, 1)      &   -       &     -     &   0, 1    &   -       &   -       &   0, 1    &  -    &   -    \\  \midrule
(1, 0)      &   0       &     0     &   0       &   0       &   1       &   1       &  1    &   1    \\
(1, 1)      &   1       &     1     &   1       &   1       &   0       &   0       &  0    &   0    \\
\bottomrule
\end{tabular}}
\label{table:carrybits}
\end{table}

The complexity of the above steps is mainly introduced by Step~(1), which involves the exploration of
all the first $(n-1)$ bit planes of $y$ in $\mathbb{G}$ to obtain $\mathbb{G}_{j}$.
It can be inferred that the computational complexity is only $O((n-1)\cdot g)$, which is much smaller than the
complexity of the previous algorithm $O(2^{n-1}\cdot g)$. Besides, this algorithm generates only a
single possible candidate $k$, thus avoiding the problem of selecting $k$ from its candidate
set\footnote{In fact, every element in $\mathbb{K}_g$ contains the same number of correct bits of
$k$ in average.} $\mathbb{K}_g$. Without loss of generality, assume that all the known parameters
$\alpha$, $\beta$ and $y$ are uniformly distributed in the interval $[0, 2^{n-1}]$.
Finally, the probability that the first $i$ ($0\leq i < n-1$) LSBs can be confirmed by
$\mathbb{G}$, denoted as $\Prob(k_{0\sim i}\mid \mathbb{G})$, is given as
\begin{equation*}
\Prob(k_{0\sim i}\mid \mathbb{G}) = \left( 1 - \left(\frac{1}{2} \right)^g \right)^{i+1}.
\end{equation*}

Assuming $n=8$ as in the three image cryptosystems studied in Sec.~\ref{sec:threeciphers},
we depict in Fig.~\ref{fig:probability} this probability with respect to different values of $g$.
As we can observe from this figure, the probability is relative high for small $i$ when $g$ equals $3$.
This result is further verified by carrying out experiments to Parvin's cryptosystem
under the assumption that the key streams $K$ is generated using the key schedule described in
\cite[Sec.~2]{parvin2014new} while
we artificially set $U$ and $V$ to zeros to fit our model proposed in Sec.~\ref{subsec:4.1}. Then, we use $2$ and $4$ known
plain-images and their corresponding cipher-images, i.e, $g=1$ and $g=3$, to recover the key stream
$K$ using the algorithm described above. The recovered key stream is used to decrypt the cipher-image of ``Baboon",
as shown in Fig.~\ref{fig:recover}b), and the deciphered results are shown
respectively in Fig.~\ref{fig:recover}c) and Fig.~\ref{fig:recover}d).


\begin{figure}[htbp]
\centering
\includegraphics[width=3.5in]{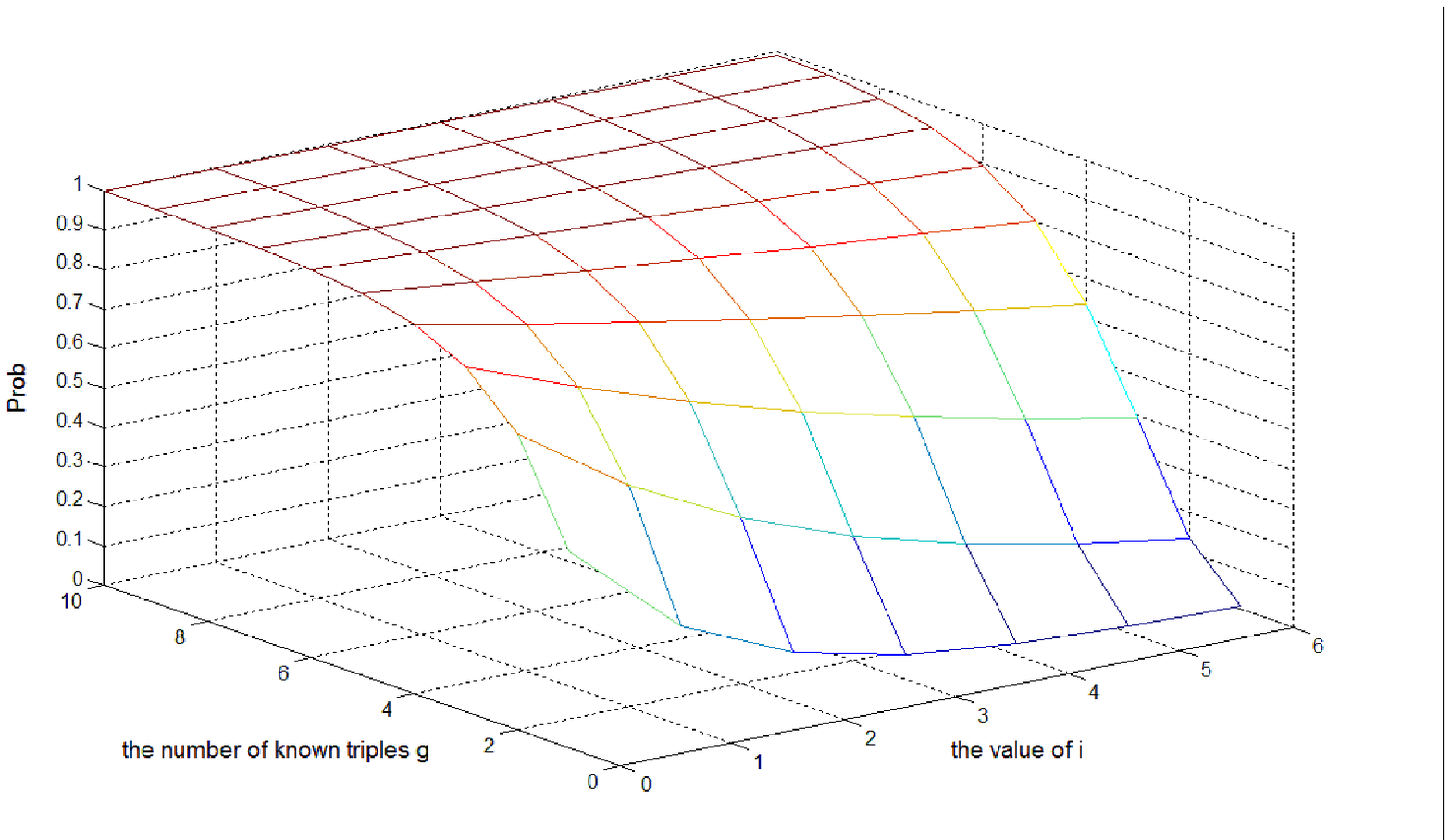}
\caption{The probability that the first $i$ LSBs of $k$ can be confirmed with respect to different $g$.}
\label{fig:probability}
\end{figure}


\begin{figure}[!htb]
\centering
\begin{minipage}{\figwidth}
\centering
\includegraphics[width=\textwidth]{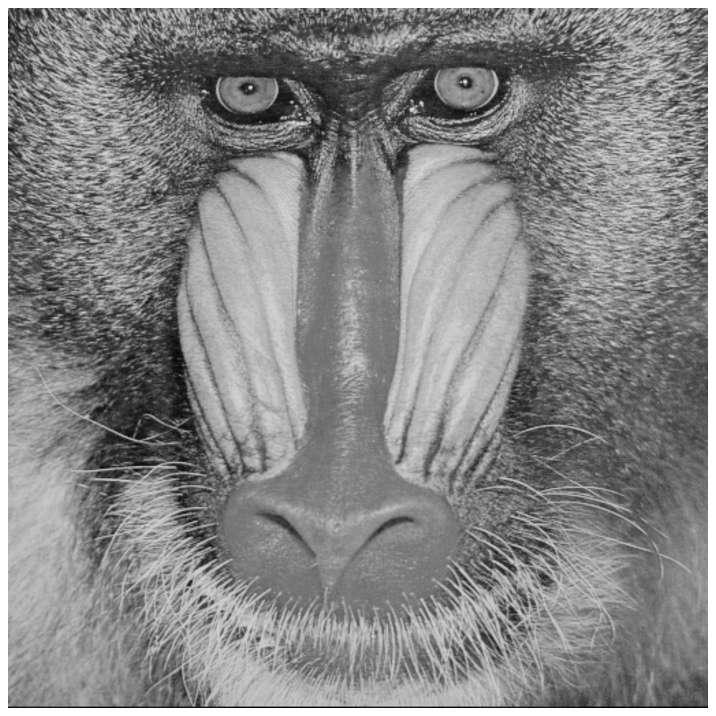}
a)
\end{minipage}
\begin{minipage}{\figwidth}
\centering
\includegraphics[width=\textwidth]{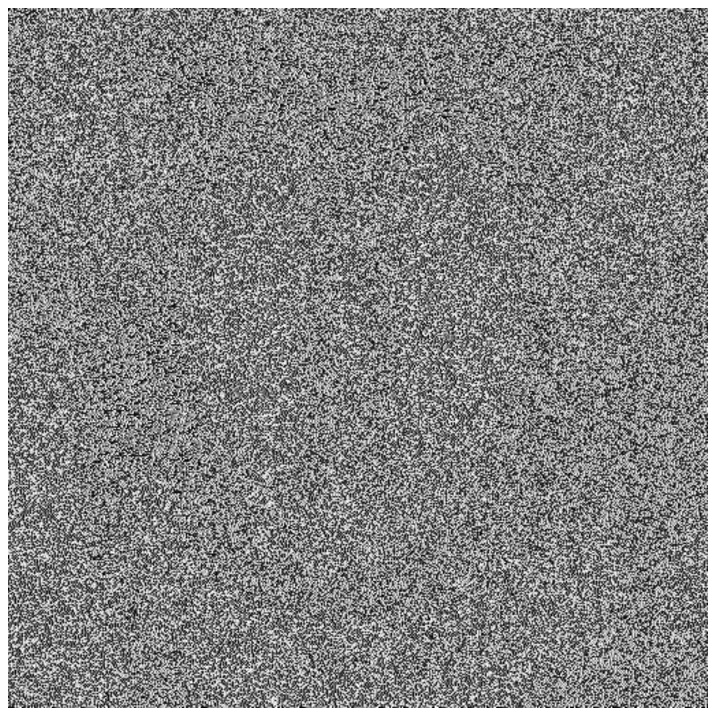}
b)
\end{minipage}
\begin{minipage}{\figwidth}
\centering
\includegraphics[width=\textwidth]{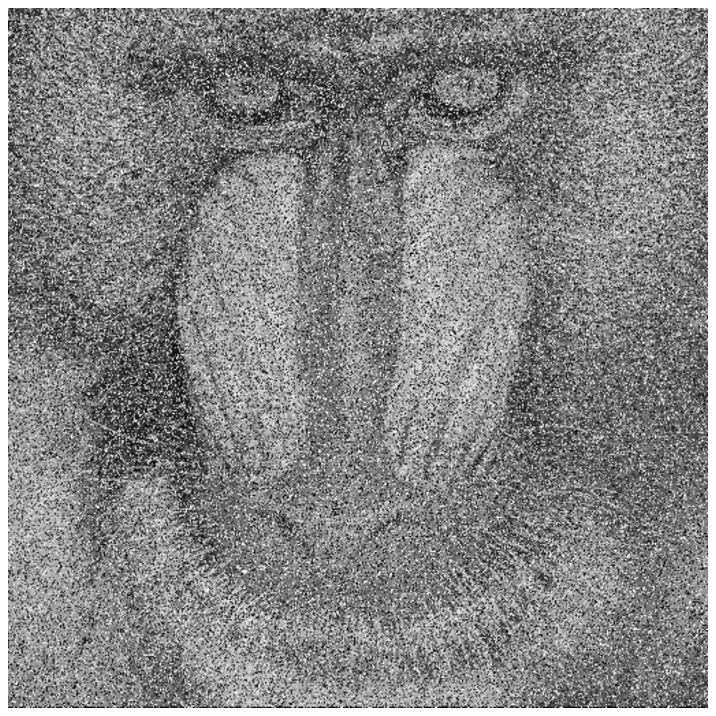}
c)
\end{minipage}
\begin{minipage}{\figwidth}
\centering
\includegraphics[width=\textwidth]{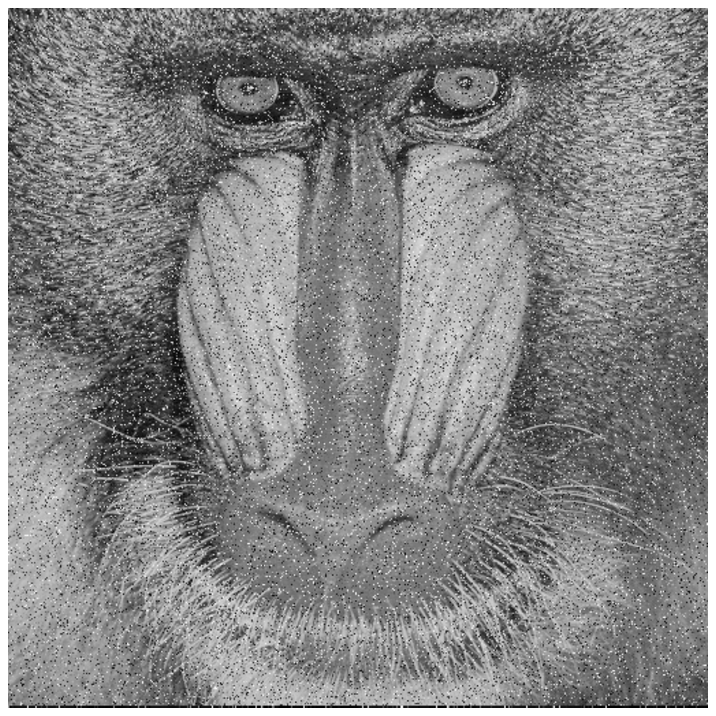}
d)
\end{minipage}
\caption{Numerical tests on simplified Parvin's cryptosystem:
a)~Plain-image ``Baboon'' of size $512 \times 512$;
b)~Encryption result of Fig~\ref{fig:recover}a) using the modified Parvin's cryptosystem;
c)~Recovery result using $2$ pairs of known plain-images and their corresponding cipher-images;
d)~Recovery result using $4$ pairs of known plain-images and their corresponding cipher-images.}
\label{fig:recover}
\end{figure}

\subsection{Cryptographic strength of the equation $(\alpha \dotplus k) \oplus g(\beta, k)=y$}
\label{subsec:5.2}
Accordingly to the results obtained in the previous section, the diffusion mechanism characterized by Eq.~(\ref{eq:model1}) is weak with respect to both
CP and KP attacks. Specifically, two groups of chosen parameters are enough to uniquely determine $k$, while a few groups
of known parameters are sufficient to determine $k$ with overwhelming probability. The bidirectional diffusion scheme
introduced in \cite{norouzi2014simple} and in \cite{yang2015novel}, and defined by Eqs.~(\ref{eq:diffusion2}) and (\ref{eq:prediffusion2}), is suggested as a workaround.
The idea of the new design is that all the pixels located after the current one are used in the diffusion process, with an avalanche effect (and so, an improvement) in the encryption of plain-images.

In the context of a CP attack scenario, thanks to the results shown in Sec.~\ref{sec:problem}, the birectional
diffusion scheme is immediately proven to be weak, since Eq.~(\ref{eq:diffusion2}) can be converted to the form of Eq.~(\ref{eq:model1}).
Considering that there are $L$ pixels in an image, the data complexity (i.e., required number of
plain-images and cipher-images) for breaking the cipher in \cite{norouzi2014simple} is only $O(L)$.

Furthermore, we can show that in the context of a KP attack scenario,
the data complexity for breaking the cipher in  \cite{norouzi2014simple} is the same as above.
Let us consider the equation
\begin{equation*}
 (\alpha \dotplus k) \oplus g(\beta, k) = y,
\end{equation*}
where $g(\beta, k) = \lfloor \beta \cdot (10^8 \cdot k) \rfloor \bmod 256$, $\alpha, y, k \in [0,255]$ and $\beta$ is
a non-negative real number.
Under the assumptions of a KP attack, i.e., that $\alpha, \beta$ and $y$ are known to the adversary, we can show that the data complexity for revealing k is only $O(1)$.
In other words, the inefficient bidirectional diffusion scheme actually does not improve the security level
of Eq.~(\ref{eq:model1}) with respect to KP attack.

We start our analysis from the trivial case $\beta \equiv 0$. Under this assumption, Eq.~(\ref{eq:model2}) is simplified to
\begin{equation*}
y = \alpha \dotplus k
\end{equation*}
since $g(\beta, k) = \lfloor \beta \cdot (10^8 \cdot k) \rfloor \bmod 256 \equiv 0$. Thus, $k$ can be calculated as $k = y \dotminus \alpha$.
For the general case $\beta > 0$, it is easy to observe that the value of $g(\beta, k)$ is sensitive to the changes of $k$. In
other words, given $\alpha$, $\beta$ and $y$, the result of $(\alpha \dotplus k) \oplus g(\beta, k)$ will be different from $y$
with an overwhelming probability even if $k$ slightly deviates from its true value. For convenience, let
$\mathbb{G} =\{(\alpha, \beta, y) \mid y = (\alpha \dotplus k) \oplus g(\beta, k) \}$ and assume $ \# \mathbb{G}= g = O(1)$.
The following procedures describe a method to determine $k$ from $\mathbb{G}$ by using this observation.
\begin{itemize}
\item{Step~(1)} Let $l=1$ and the solution set $\mathbb{K}_l = \emptyset$.
\item{Step~(2)} Select the $l$-th element of $\mathbb{G}$ and exhaustively test all the $2^8$ possible values of $k$ to
    check whether it satisfies Eq.~(\ref{eq:model2}). Collect all the possible values of $k$ that meet the requirement
    and denote them as $\mathbb{K}_l$.
\item{Step~(3)} Go to Step~(5) if $\# \mathbb{K}_{l} =1$.
\item{Step~(4)} Set $l=l+1$ if $l< g$. Go to Step~(2) and update the solution set by $\mathbb{K}_{l+1} = \mathbb{K}_{l+1} \cap \mathbb{K}_{l}$.
\item{Step~(5)} Print the value of the single element of $\mathbb{K}_{l}$ if $\# \mathbb{K}_{l}=1$.
    Otherwise output $\# \mathbb{K}_{l}$.
\end{itemize}

We verify the validity of this algorithm by carrying out experiments to Norouzi's cryptosystem (that can be viewed
as the simplified version of Yang's design). Three $512 \times 512$ known plain-images with different statistical
characteristics are employed as our test images (Fig.~\ref{fig:decrypt2}a)-c)).
These images are encrypted using Norouzi's cryptosystem under the secret key that was adopted in \cite[Sec.~3]{norouzi2014simple}.
Using the techniques illustrated in Sec.~\ref{sec:problem}, we cast the relationship between the plaintext pixels and ciphertext pixels to
the form of Eq.~(\ref{eq:model2}). Then, we respectively use $1$, $2$ and $3$ pairs of plain-images and their corresponding cipher-images
to retrieve the equivalent secret key stream $K$ by the above algorithm.
\begin{figure}[!htb]
\centering
\begin{minipage}{\figwidth}
\centering
\includegraphics[width=\textwidth]{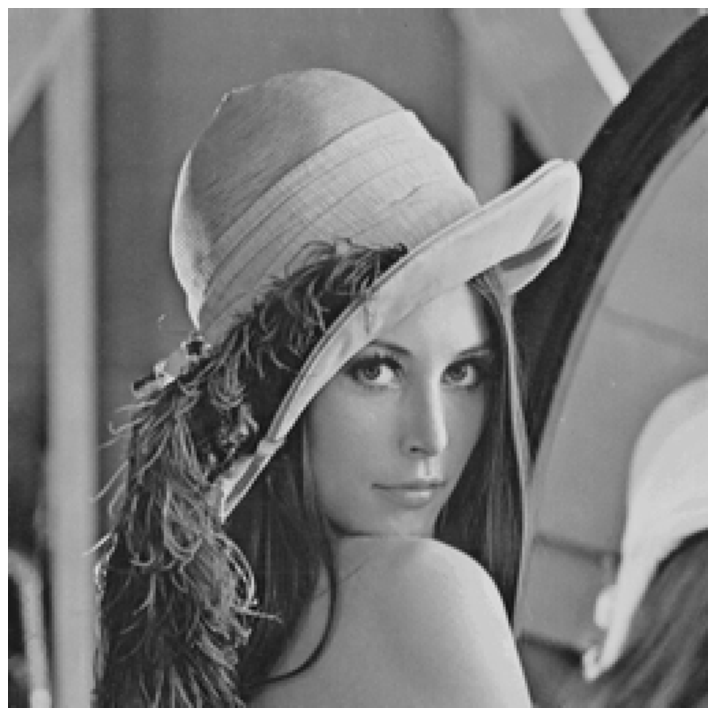}
a)
\end{minipage}
\begin{minipage}{\figwidth}
\centering
\includegraphics[width=\textwidth]{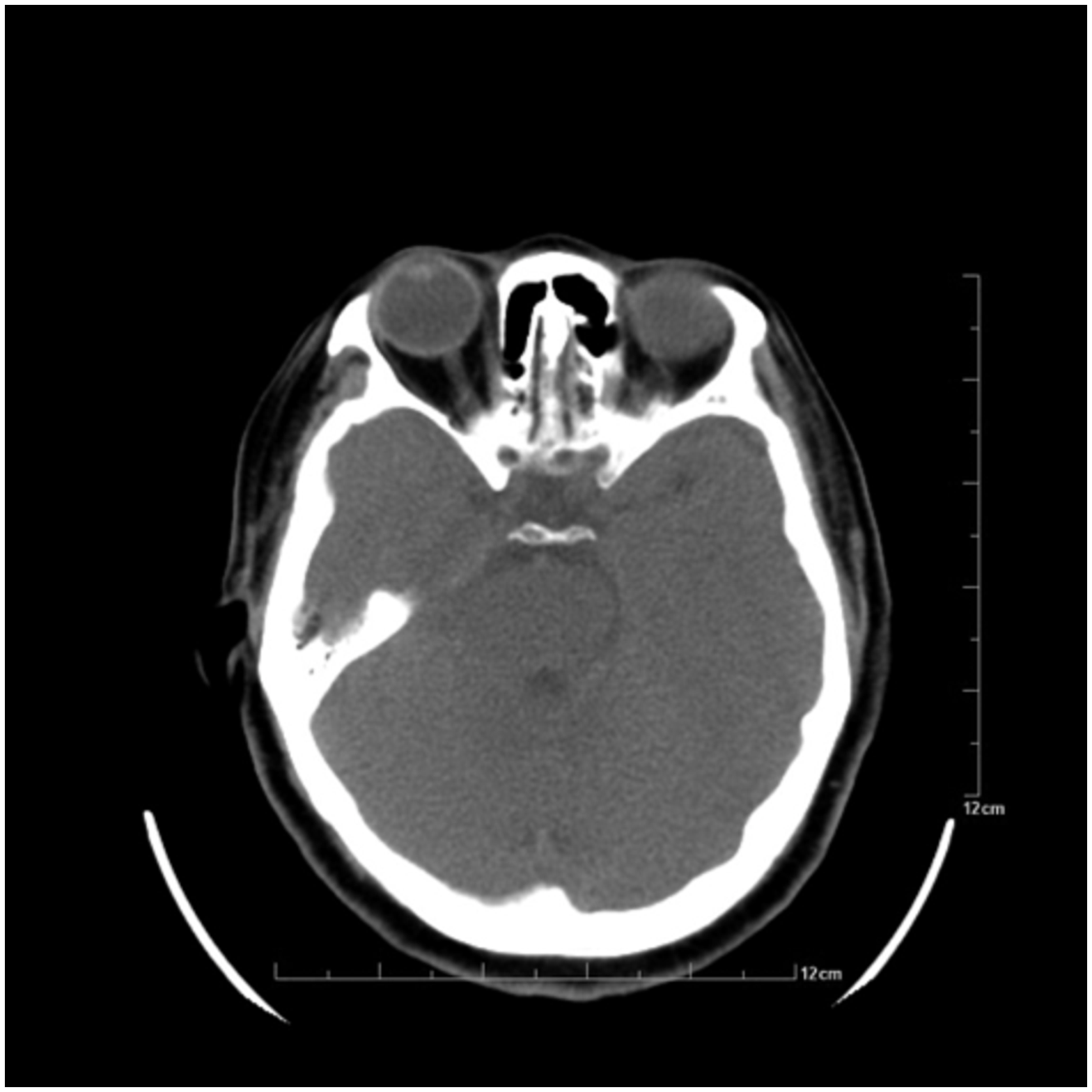}
b)
\end{minipage}
\begin{minipage}{\figwidth}
\centering
\includegraphics[width=\textwidth]{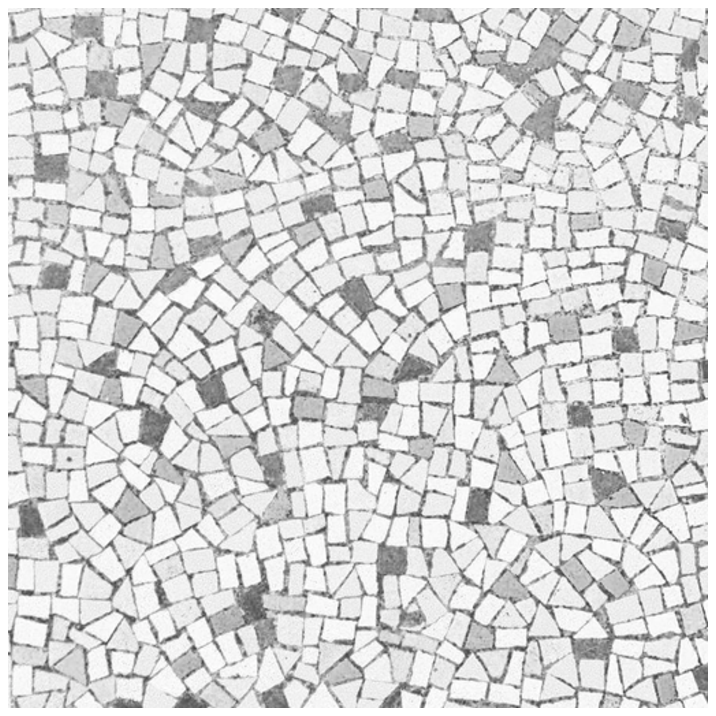}
c)
\end{minipage}
\caption{Three test images for recovering the equivalent key stream of Norouzi's cryptosystem:
a)~``Lena'';
b)~CT image;
c)~Mosaic image.}
\label{fig:decrypt2}
\end{figure}
The average \textit{recovery rates} of the proposed KP attack using different numbers of known plain-images are listed in Table~\ref{table:RR}.
Here, the \textit{recovery rate} is defined as
\begin{equation*}
\textit{recovery rate} = \frac{\textit{number of correctly recovered elements of~} K}{\textit{total number of elements in~} K} \times 100\%.
\end{equation*}
It can be observed that the average \textit{recovery rate} raises as the number of known plain-images increase.
Even the number of known plain-images is only $1$, the average \textit{recovery rate} is close to $67\%$. When
the number of known plain-images is $3$, the \textit{recovery rate} grows to $100\%$. Furthermore, we utilize these
recovered equivalent key streams to decrypt an intercepted cipher-image and the result is shown in Fig.~\ref{fig:decrypt3}a)-c).
From Fig.~\ref{fig:decrypt3}, it is concluded that $100\%$ \textit{recovery rate} of the key stream guarantees perfect
reconstruction of the intercepted cipher-image, while a high \textit{recovery rate} of the key stream does not lead to
good or acceptable visual quality. This phenomenon is attributable to the bidirectional diffusion property of Eq.~(\ref{eq:multiply}),
where the error of a wrongly decrypted pixel will spread to all successive decryption in a pseudo-random manner.

\begin{table}
\centering
\caption{Average recovery rate using different numbers of known plain-images.}
\centering
    \begin{tabular}{cc}\bottomrule
    Number of known plain-images & average \textit{recovery rate} \\ \midrule
    1                            & 66.6637\%                    \\
    2                            & 99.8247\%                    \\
    3                            & 100\%                    \\
    \bottomrule
    \end{tabular}
    \label{table:RR}
\end{table}

\begin{figure}[!htb]
\centering
\begin{minipage}{\figwidth}
\centering
\includegraphics[width=\textwidth]{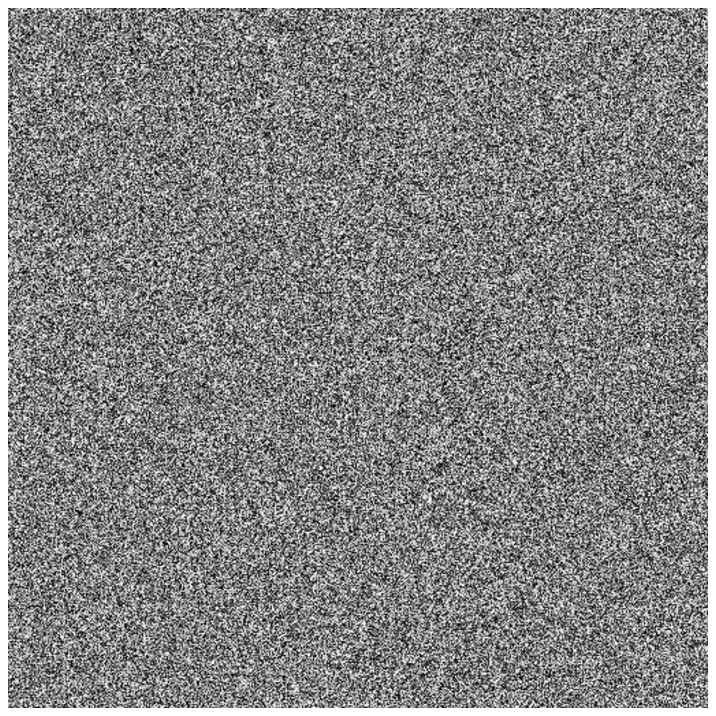}
a)
\end{minipage}
\begin{minipage}{\figwidth}
\centering
\includegraphics[width=\textwidth]{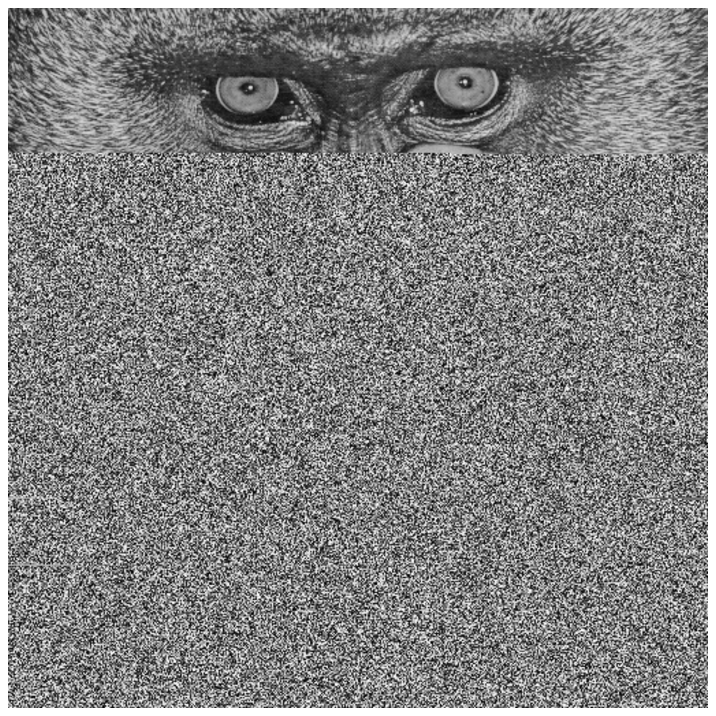}
b)
\end{minipage}
\begin{minipage}{\figwidth}
\centering
\includegraphics[width=\textwidth]{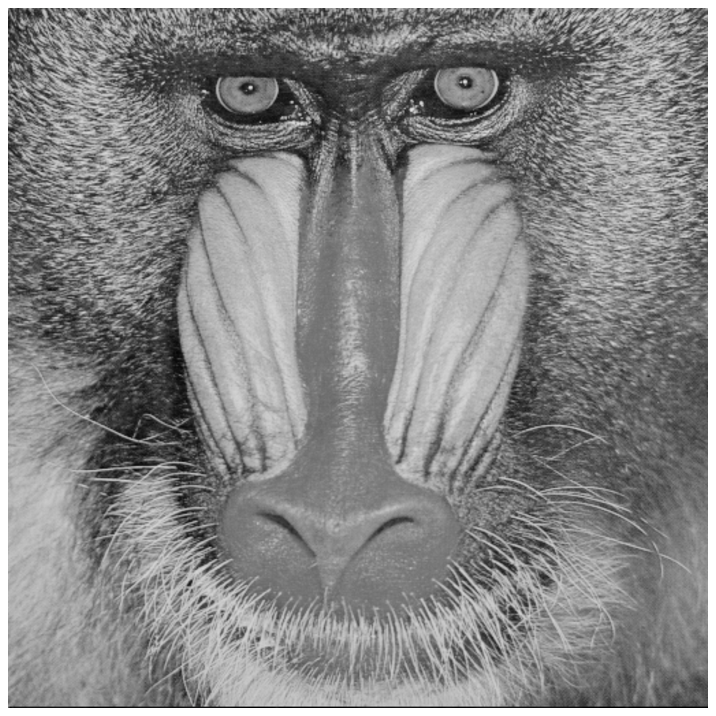}
c)
\end{minipage}
\caption{Recovery results:
a)~Deciphered result using the key stream retrieved from Fig.~\ref{fig:decrypt2}a);
b)~Deciphered result using the key stream retrieved from Fig.~\ref{fig:decrypt2}a) and b);
c)~~Deciphered result using the key stream retrieved from Fig.~\ref{fig:decrypt2}a)-c).}
\label{fig:decrypt3}
\end{figure}

\section{Cryptographic applications}
\label{sec:applications}
Exploiting the security analyses of Eq.~(\ref{eq:model1}) and Eq.~(\ref{eq:model2}) shown above, this section presents
plaintext attacks to the full cryptosystems proposed in \cite{parvin2014new,norouzi2014simple,yang2015novel} and briefly discusses
other security implications related to our analyses.

\begin{enumerate}[A.]
\item{\textbf{Cryptanalysis of Parvin's cryptosystem}}

As described in Sec.~\ref{sec:threeciphers}, Parvin's cryptosystem is composed of circular permutations and a single diffusion stage.
To apply our analysis result presented in Sec.~\ref{subsec:5.1}, we need first to recover the equivalent key streams used for
row and column circular permutation. The underlying strategy is to study the relationship between cipher-images produced by some
some bottom-line chosen plain-images whose elements are invariant with respect to row and column permutations. Similar ideas are
also employed to analyze other chaos-based cryptosystems \cite{wang2014cryptanalysis,arroyo2013cryptanalysis,li2012breaking}.
Here, we suppose that an image having fixed gray value is available and denote it as $P_1=\{p_1(i, j) \equiv 0\} _{i=1, j=1}^{H, W}$.
Then, we set $p_1(1, 1)=128$ and keep all the other pixels unchanged and denote the modified image by $P_2=\{p_2(i, j)\}_{i=1, j=1}^{H, W}$.
Figure~\ref{fig:decryptParvin}a) and b) depict the cipher-images corresponding to $P_1$ and $P_2$, respectively. Here, $H = W = 512$ is
chosen. The difference of the two cipher-images is shown in Fig.~\ref{fig:decryptParvin}c). Find the first pixel whose value
is $128$ and denote its position by $(i_1, j_1)$. Referring to Eqs.~(\ref{eq:parvin1}), (\ref{eq:pavin2}) and (\ref{eq:diffusion}),
it can be concluded that $u(1) = ((j_1-1) \bmod H) +1$ and $v(1) = ((i_1-1) \bmod W) +1$. Repeat this test for all the diagonal pixels of $P_1$,
$U$ and $V$, the key streams for row and column permutations, can be retrieved completely. Combining with the analysis presented
in Sec.~\ref{subsec:5.1}, the data complexity of the CP attack is $O(1)+\max(H, W)$ with an overwhelming probability.
\begin{figure}[!htb]
\centering
\begin{minipage}{\figwidth}
\centering
\includegraphics[width=\textwidth]{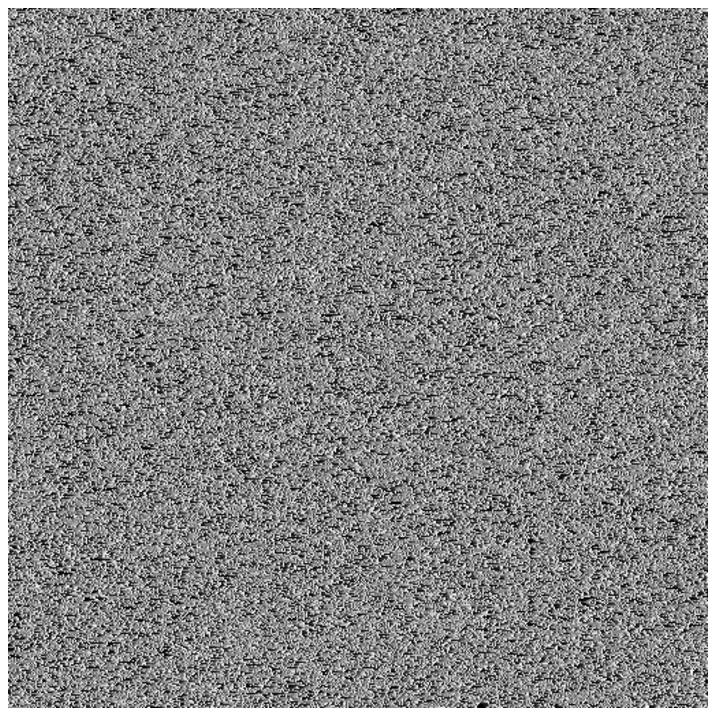}
a)
\end{minipage}
\begin{minipage}{\figwidth}
\centering
\includegraphics[width=\textwidth]{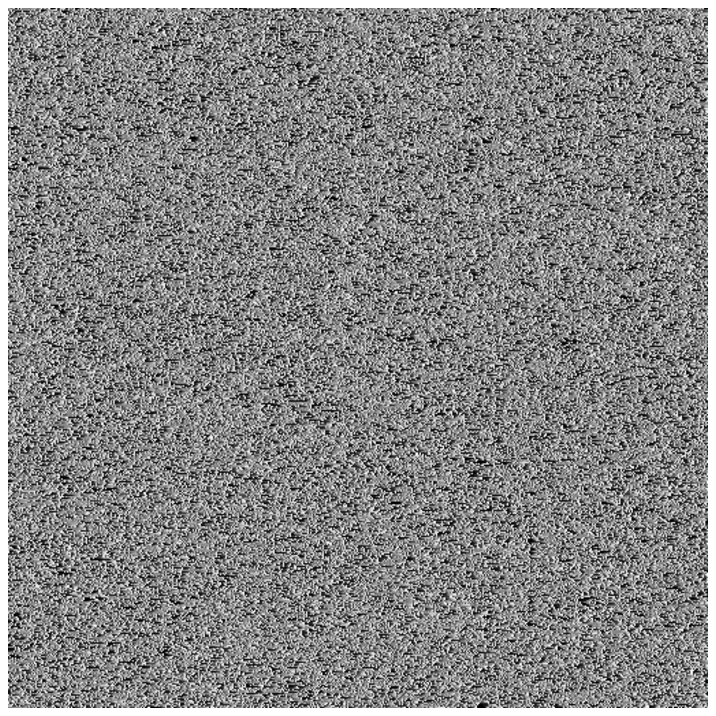}
b)
\end{minipage}
\begin{minipage}{\figwidth}
\centering
\includegraphics[width=\textwidth]{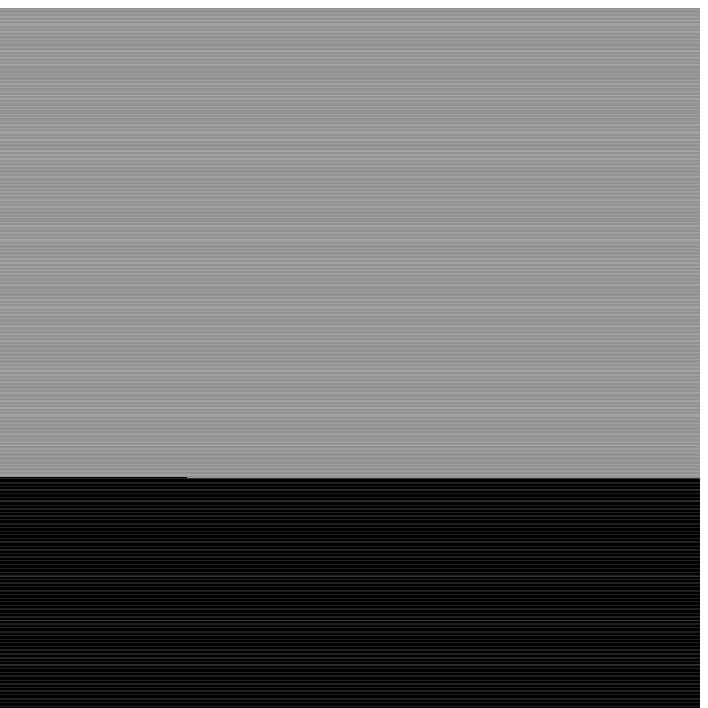}
c)
\end{minipage}
\caption{Example test for recovering the equivalent permutation key streams of Parvin's cryptosystem:
a) The cipher-image of $P_1$;
b) The cipher-image of $P_2$;
c) The difference between Figs.~\ref{fig:decryptParvin}a) and b) using XOR operation.}
\label{fig:decryptParvin}
\end{figure}

\item{\textbf{Cryptanalysis of Norouzi's and Yang's cryptosystems}}

Applying the analysis presented in Sec.~\ref{subsec:5.2}, it is readily to conclude that Norouzi's cryptosystem can be
compromised in KP attack scenario at data complexity $O(1)$. For Yang's scheme, the remaining task is to recover the
remaining key streams used for permutation. By noting that Yang's scheme is different from Parvin's only by the order
of diffusion and permutation in the present context, we use the similar strategy to reveal the equivalent permutation
key streams of Yang's cryptosystem. For example, to reveal $v(H)$ and $u(W-2)$, we employ three chosen-images $P_1$, $P_2$
and $P_3$ with the form
\begin{IEEEeqnarray}{rCl}
P_1  &=& [0, 0, 0, \cdots, 0, 0, 0, 1], \nonumber \\
P_2  &=& [0, 0, 0, \cdots, 0, 0, 1, 0], \nonumber \\
P_3  &=& [0, 0, 0, \cdots, 0, 1, 0, 0]. \nonumber
\end{IEEEeqnarray}
According to Eqs.~(\ref{eq:Yang1}), (\ref{eq:Yang2}) and (\ref{eq:Yang3}), their corresponding cipher-images $C_1$, $C_2$ and $C_3$
satisfy the following two conditions: 1) there are two distinct ciphertext elements between $C_1$ and $C_2$, 2) there are
three distinct ciphertext elements between $C_3$ and $C_1$ (or $C_2$). Comparing $C_1$, $C_2$ and $C_3$, the location
of $c_1(H, W-2)$ can be identified. Figure~\ref{fig:illustration} sketches the rules involved in this procedure.
Repeat this test to the last row and column of $P_1$, the equivalent permutation key streams $U$ and $V$ can be fully recovered
at the data complexity\footnote{The permutation for the last two pixels can be retrieved by brute force search.}
$O(H+W)$ under CP attack.

\begin{figure}[htbp]
\centering
\includegraphics[width=6in]{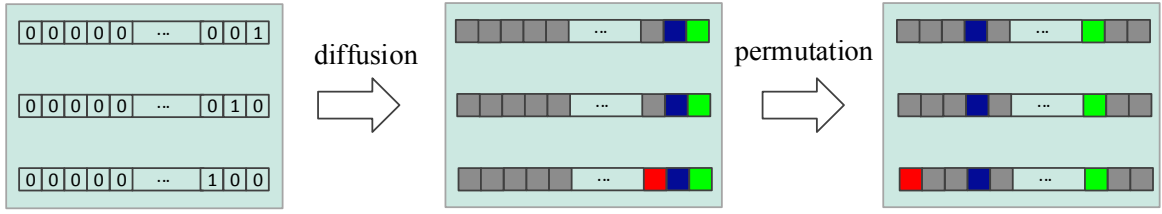}
\caption{Illustration of the CP attack on Yang's cryptosystem to recover the equivalent secret key used for permutation.}
\label{fig:illustration}
\end{figure}

\item{\textbf{Other cryptographic implications}}

Observing that the analysis with respect to the equation $ (\alpha \dotplus k) \oplus g(\beta, k) = y$ involves exhaustive
searching the possible key space, an intuitive workaround for Norouzi's and Yang's cryptosystems is to group several pixels as a
single element to enlarge the real key space. For example, combine $15$ pixels together will make the key space grows
to $2^{120}$ and frustrate the KP attack presented in Sec.~\ref{subsec:5.2}. However, Norouzi's and Yang's cryptosystems can
be cast to the form of $(\alpha \dotplus k) \oplus (\beta \dotplus k) = y$ in CP attack scenario and cryptanalysis of this
equation is regardless of the bit length of the plaintext. It can be concluded that using composite pixel representation as a
remedy is futile.

Regarding the widely usage of the diffusion equation~(\ref{eq:diffusion}) \cite{fridrich1998symmetric,zhu2012novel,eslami2013improvement,zhang2011novel,chen2004symmetric,mao2004novel,liu2011color,
zhou2015cascade}, our analysis on the equation $(\alpha \dotplus k) \oplus (\beta \dotplus k) = y$ seems useful in evaluating
security of other ciphers also based on this kind of diffusion mechanism. The fact that the search space of the unknown $k$
could be reduced from $2^{2n}$ to $O(1)$ indicates
that a loophole exists in the corresponding crytosystems, and that it can be used to retrieve information
about the key.
Even worse, this loophole cannot be fixed by choosing a larger $n$. With this concern, we recommend
using some relative strong diffusion schemes with respect to KP and CP attacks, such as
$(k_1 \dotplus k_2)\oplus (k_1 \dotplus (k_2 \oplus \beta)) = y$ \cite{paul2005solving}.
\end{enumerate}

\section{Conclusion}
\label{sec:conclusion}
Considering the three cryptosystems proposed in \cite{parvin2014new,norouzi2014simple,yang2015novel} as case studies,
we have studied the security properties of equations (i) $(\alpha \dotplus k) \oplus (\beta \dotplus k) = y$  and (ii)
$(\alpha \dotplus k) \oplus g(\beta, k)=y$. The underlying theory of the key scheduling process employed in these example
crytosystems ranging from chaotic/hyper-chaotic function to quantum computation, which are regarded as having different
characteristics. However, our analyses reveal that all the three ciphers are very weak upon plaintext attacks.
Specifically, the equivalent key streams used in these designs can be retrieved using a small number of plain-images.
We provide a sufficient condition to determine the unknown $k$ of equation (i) under the KP attack scenario. The
relationship of our result and the existing ones under CP attack assumption \cite{li2011breaking,li2013breaking,wang2014cryptanalysis}
is also investigated. The algorithms provided and the extensive numerical experiments confirm that both equation (i) and (ii) can
be solved using only $O(1)$ known plaintexts. In this concern, it is readily to conclude that most image ciphers based
on a single round permutation-diffusion architecture are insecure with respect to plaintext attacks. Our work can be extended
to investigate diffusion equations involves more complex cryptographic primitives, such as modulo multiplication \cite{Zhang2015breaking}.

\section*{Acknowledgements}

This research was partly supported by the Research Activities Fund of City University of Hong Kong and
Fundamental Research Funds for the Central Universities (XDJK2015C077).

\appendix

\section{Proof of Theorem~2}
\label{sec:proof}
Let us consider the equivalent form of Eq.~(\ref{eq:model1}), i.e.,
\begin{equation}
\tilde{y} = (\alpha \dotplus k) \oplus (\beta \dotplus k) \oplus \alpha \oplus \beta.
\label{eq:equalform}
\end{equation}
Observe that the $(i+1)$-th bit of $\tilde{y}$, i.e., $\tilde{y}_{i+1}$, can be calculated using only the previous bits $\alpha_i$,
$\beta_i$, $k_i$, $c_i$, $\tilde{c}_i$, ($i \in [0, n-2]$) by the following three equations
\begin{equation}
\left\{
\begin{IEEEeqnarraybox}[][c]{rCl}
\IEEEstrut
\tilde{y}_{i+1} &=& {c}_{i+1} \oplus \tilde{c}_{i+1} , \\
c_{i+1} &=& k_i\alpha_i \oplus k_ic_i \oplus \alpha_ic_i , \\
\tilde{c}_{i+1} &=&k_i\beta_i \oplus k_i\tilde{c}_i \oplus \beta_i\tilde{c}_i,
\IEEEstrut
\end{IEEEeqnarraybox}
\right.
\label{eq:bitform}
\end{equation}
where $c_i$ is the carry bit at the $i$-th bit plane of $(\alpha \dotplus k)$ and $\tilde{c_i} = \tilde{y_i} \oplus c_i$.
Table~\ref{table:tildeyi} lists the values of $\tilde{y}_{i+1}$ that computed from Eq.~(\ref{eq:bitform}) under all the possible values
of $\alpha_i, \beta_i, \tilde{y}_i, k_i$ and $c_i$.
\begin{table}[htp!]
\centering
\caption{The values of $\tilde{y}_{i+1}$ corresponding to the values of $\alpha_i, \beta_i, \tilde{y}_i, k_i$ and $c_i$.}
\centering
{\begin{tabular}{*{8}{c}c}
\toprule
\multirow{2}{*}{$(k_{i}, {c}_{i})$}     & \multicolumn{8}{c}{$(\alpha_i, \beta_i, \tilde{y}_{i} )$} \\
\cline{2-9} & $(0,0,0)$ &$(0,0,1)$& $(0,1,0)$ & $(0,1,1)$& $(1,0,0)$& $(1,0,1)$  &$(1,1,0)$   &$(1,1,1)$ \\\midrule
(0, 0)      &     0     &   0     &   0       &   1      &   0      &   0        &      0     &  1       \\
(0, 1)      &     0     &   0     &   1       &   0      &   1      &   1        &      0     &  1       \\  \midrule
(1, 0)      &     0     &   1     &   1       &   1      &   1      &   0        &      0     &  0       \\
(1, 1)      &     0     &   1     &   0       &   0      &   0      &   1        &      0     &  0       \\ \hline
            &   Col(1)  &    Col(2) &   Col(3)  &   Col(4)  &   Col(5)  &   Col(6)  & Col(7)&  Col(8)\\
\bottomrule
\end{tabular}}
\label{table:tildeyi}
\end{table}

Table~\ref{table:tildeyi} indicates that $k_i$ can be determined if $(\alpha_i, \beta_i, \tilde{y}_{i})$ falls in
\{Col(2), Col(3), Col(5), Col(8)\}, i.e., $y_i = \tilde{y}_i \oplus \alpha_i \oplus \beta_i =1$, and $c_i$ is known.
Based on this observation, the theorem can be proved by mathematical induction on $i$ ($0 \leq i \leq n-2$). We first
consider the case for $i=0$. Since $c_0 \equiv \tilde{c}_0 \equiv 0$, the condition
\begin{IEEEeqnarray}{rCl}
y_0 &=& \tilde{y}_0 \oplus \alpha_0 \oplus \beta_0 \nonumber \\
  &=& c_0 \oplus \tilde{c}_0 \oplus \alpha_0 \oplus \beta_0 \nonumber \\
  &=& \alpha_0 \oplus \beta_0 \nonumber\\
  &=& 1 \nonumber
\end{IEEEeqnarray}
implies
\begin{IEEEeqnarray}{rCl}
\tilde{y}_1 &=& c_1 \oplus \tilde{c}_1 \nonumber \\
            &=& k_0 \alpha_0 \oplus k_0 \beta_0 \nonumber \\
            &=& k_0(\alpha_0 \oplus \beta_0) \nonumber\\
            &=& k_0. \nonumber
\end{IEEEeqnarray}
Hence the theorem is proved for the case $i=0$. Assume that it is valid for $i = m$ ($m \leq n-3$), i.e.,
all the $m$ least significant bits of $k$ are confirmed when $y = \sum\nolimits_{j=0}^{m-1} 2^m$ and thus all the $c_i$ and $\tilde{c}_i$
can be derived by Eq.~(\ref{eq:bitform}) for all $i \in [0, m+1]$. Then, for the case $i=m+1$, the condition $y_{m+1} =1$  implies that
\begin{IEEEeqnarray}{rCl}
{y}_{m+1} = c_{m+1} \oplus \tilde{c}_{m+1} \oplus \alpha_{m+1} \oplus \beta_{m+1}   = 1 \nonumber
\end{IEEEeqnarray}
holds when referring to Eqs.~(\ref{eq:equalform}) and (\ref{eq:bitform}). When computing $y_{m+2}$ by Eq.~(\ref{eq:bitform}),
we have
\begin{IEEEeqnarray}{rCl}
\tilde{y}_{m+2} &=& c_{m+2} \oplus \tilde{c}_{m+2} \nonumber \\
                &=& k_{m+1} \alpha_{m+1} \oplus k_{m+1} \beta_{m+1} \oplus k_{m+1} c_{m+1} \oplus k_{m+1} \tilde{c}_{m+1}
                \oplus \alpha_{m+1} c_{m+1} \oplus \beta_{m+1} \tilde{c}_{m+1}
                \nonumber \\
                &=&  k_{m+1}(\alpha_{m+1} \oplus \beta_{m+1}\oplus c_{m+1} \oplus \tilde{c}_{m+1})
                \oplus \alpha_{m+1} c_{m+1} \oplus \beta_{m+1} \tilde{c}_{m+1} \nonumber\\
                &=& k_{m+1} \oplus \alpha_{m+1} c_{m+1} \oplus \beta_{m+1} \tilde{c}_{m+1}.  \nonumber
\end{IEEEeqnarray}
Observing that $\alpha_{m+1}$, $\beta_{m+1}$ and $\tilde{y}_{m+2}$ are known parameters in our KP attack scenario,
$c_{m+1}$ and $\tilde{c}_{m+1}$ are the result from the previous induction step, we conclude that
\begin{IEEEeqnarray}{rCl}
{k}_{m+1} = \tilde{y}_{m+2}  \oplus \alpha_{m+1} c_{m+1} \oplus \beta_{m+1} \tilde{c}_{m+1}, \nonumber
\end{IEEEeqnarray}
thus completing the mathematical induction and hence proving the theorem.



\bibliographystyle{elsarticle-num}
\bibliography{ref}







\end{document}